\RequirePackage{fixltx2e}

\pdfoutput=1
\documentclass[
submission
, nomarks
]{dmtcs-episciences}

\usepackage[T1]{fontenc}
\usepackage[utf8]{inputenc}
\usepackage{amsmath}
\usepackage{amsthm}
\usepackage{hyperref}

\makeatletter
\theoremstyle{plain}
\newtheorem{thm}{\protect\theoremname}
  \theoremstyle{definition}
  \newtheorem{defn}[thm]{\protect\definitionname}
  \theoremstyle{remark}
  \newtheorem*{rem*}{\protect\remarkname}
  \theoremstyle{plain}
\newtheorem*{lem*}{\protect\lemmaname}
 
  \theoremstyle{plain}
\newtheorem*{thm*}{Main Theorem}

\newcounter{lemma}

\usepackage{amsfonts} 
\usepackage{hyperref}

\makeatother

  \providecommand{\definitionname}{Definition}
  \providecommand{\lemmaname}{Lemma}
  \providecommand{\remarkname}{Remark}
  \providecommand{\theoremname}{Theorem}
\providecommand{\theoremname}{Theorem}

\def\nomath#1{\ifmmode\text{#1}\else#1\fi}

\def\interval#1#2{\left\langle#1,#2\right\rangle}

\def\eqspaces{\;=\;}
\def\lespaces{\;\le\;}

\def\Prob#1{\Pr\!\left[#1\right]}
\def\ProbParam#1#2{\Pr_{#1}\!\left[#2\right]}

\def\OO{\mathcal O}

\def\natural{\mathbb{N}}

\def\emphDef#1{{\bf#1}}
\def\emphDefMath#1{\mathbf{#1}}

\title{Predecessor problem on smooth distributions}

\affiliation{
  Faculty of Mathematics and Physics, Charles University in Prague, Czech Republic \\
  CZ.NIC, z.\,s.\,p.\,o.
}

\author{Vladim\'ir \v{C}un\'at\affiliationmark{1,2}%
  \thanks{This research was supported by the Czech Science Foundation under grant GA14-10799S.}
}

\received{2017-01-25}

\begin{document}

\maketitle

\begin{abstract}
We follow a research thread studying the predecessor problem on input
taken from ``smooth'' distributions family. We propose a conceptually
simpler solution, utilizing well-known results from much better studied
variant of the problem that assumes nothing about the input. As a
side effect, we are able to extend the range of input distributions
handled in $\mathcal{O}(\log\log n)$ time, which are the most
studied cases, and we provide better insight into why the related
methods are faster on smooth inputs.
\end{abstract}



\section{Introduction}

The predecessor problem involves maintaining a set of linearly
ordered keys and performing queries on that set. The basic variant allows insertions, deletions, and -- more complex to perform -- \emph{predecessor} query: find the largest contained key that is less than a given value. 

This paper investigates the predecessor problem on inputs described by  a~well studied class of ``smooth'' probability distributions.
In the following section we formalize models that we use for input and computation,
and we introduce the class of smooth distributions. In~Section~\ref{sec:smooth-bucket} we
prove a key property of smooth distributions that seems to be
implicitly used in all related work (in weaker forms) as the most
important step to achieve the stated performance. In~Section~\ref{sec:main-results}
we show how to utilize this property directly with well-known distribution-independent
structures for the predecessor problem to get better performance in
a simpler way, essentially by converting to the case of polynomial-sized
universe.

\section{Definitions \label{sec:defs}}

We assume the usual word-RAM~\cite{HagerupT98}
as our computational model, and additionally we restrict keys to word-sized
integers. That can be shown not to be a~significant limitation, as
many other key types can be converted to the integer case cheaply,
for example standard floating-point numbers \cite[sec. 2.1.3]{Goldberg91}
and strings \cite{AndersT01}.

We study the behavior of the predecessor problem in cases where the input
has certain specific properties, which requires us to specify how we model the input:
\begin{defn}[input model]
Insertions take keys distributed according to a~density~$\mu$ that
does not change during the whole life of the structure. Deletions
remove uniformly from the set of \emph{contained} elements. All the
operations are arbitrarily intermixed and mutually independent.
\end{defn}
This model is convenient because it preserves distribution of the
stored set~– at any point of execution the set appears like a sample
chosen anew according to $\mu$. That is usually essential for simplifying complexity
analyses in related structures. Various random deletion models were
thoroughly studied by Knuth~\cite{Knuth77}.

Most of related papers work with key distributions that can be described
as \emph{$\left(s^{\alpha},s^{1-\delta}\right)$-smooth} for some
constants $\alpha,\delta>0$. Typically it is assumed that the particular
density of the distribution is not known, but $\alpha$ and $\delta$
are fixed known parameters.
\begin{rem*}
If the distribution \emph{was} known and its inverse cumulative distribution
function $F^{-1}$ was cheaply computable (or approximable), we could
convert the problem to the uniform case (or another bounded density).
We could simply store keys transformed by $F^{-1}$, which would preserve
the order. With any bounded input distribution, all operations can
be easily handled in constant time by simply splitting the the key
domain into $\Theta(n)$ intervals and mapping each to an element
of an array.\footnote{We use $n$ to denote the current size of the stored set, following
the custom in data structures.} This solution for bounded distributions was pointed out already by
Andersson and Mattsson \cite[sec. 5.2]{AndersM93}.
\end{rem*}
The concept of a smooth probability density was originally introduced by Mehlhorn, Tsakalidis \cite{MehlT93} and later generalized by Andersson, Mattsson \cite{AndersM93}. We only amend the definition given in \cite{AndersM93,KMSTTZ06} by allowing $c_2 = c_3$.%
\begin{defn}
 \label{def:smooth}
	Let $\mu$ be the density of a continuous random variable $X$ defined over an interval~$\interval a b$. Given two functions $f_1$ and $f_2$, $\mu$ is called \emphDef{$\emphDefMath{(f_1,f_2)}$-smooth} if
	\begin{multline*}
		\exists\beta \; \forall c_1,c_2,c_3 \quad a \le c_1 < c_2 \le c_3 \le b \quad \forall s \in \natural \\
		\ProbParam{X \sim \mu}{  X \in \interval{ c_2 - \frac{c_3-c_1}{f_1(s)} }{c_2} \; \middle| \; X \in \interval{c_1}{c_3}  }
		\lespaces \frac{\beta f_2(s)}{s}.
	\end{multline*}
\end{defn}

There is quite a simple intuition behind this complicated condition. It implies that if we cut some interval $\interval{c_1}{c_3}$ into $f_1(s)$ subintervals and generate $s$ values according to density $\mu$ restricted to $\interval{c_1}{c_3}$, every subinterval is expected to get $\OO(f_2(s))$ elements.

\begin{rem*}
This definition, as written, makes only sense for \emph{continuous} probability
distributions, yet any realistic models of computation can
\emph{not} assume handling arbitrary non-discrete values in constant
space and time per operation. Therefore we assume implicit rounding
when inserting a key, thus converting to the nearest value representable
in a single RAM word. It seems reasonable that the related research
assumed something similar without stating it explicitly.
\end{rem*}


\section{Static analysis of bucketing}
 \label{sec:smooth-bucket}
In this section we show how smoothness implies that an arbitrary sufficiently
short interval is expected to get only a~constant number of input keys.
\begin{lem*}
\refstepcounter{lemma}
\label{lem:smooth-bucket}
Given $\alpha,\delta>0$ and a positive
integer $n$, let us independently draw $n$ keys from a $\left(s^{\alpha},s^{1-\delta}\right)$-smooth
distribution, and split its whole domain into at least $n^{\alpha/\delta}$
equally long intervals. Then the expected number of keys in an~interval
is $\OO(1)$.
\end{lem*}
\begin{proof}
The smoothness (Definition~\ref{def:smooth}) over the domain $\interval a b$ gives us:
\begin{multline*}
	\exists\beta \; \forall c_1,c_2,c_3 \quad a \le c_1 < c_2 \le c_3 \le b \quad \forall s \in \natural \\
		\Prob{  X \in \interval{ c_2 - \frac{c_3-c_1}{s^\alpha} }{c_2} \; \middle| \; X \in \interval{c_1}{c_3}  }
		\lespaces \frac{\beta s^{1-\delta}}{s}  \eqspaces  \beta s^{-\delta}.
\end{multline*}
We choose to cover the whole domain $\interval{c_1}{c_3} := \interval a b$. The conditioning can be removed because it is always fulfilled.
\[
	\exists\beta \; \forall c_2 \quad a < c_2 \le b \quad \forall s \in \natural \quad
		\Prob{  X \in \interval{ c_2 - \frac{b-a}{s^\alpha} }{c_2}  }
		\lespaces \beta s^{-\delta}
\]

Now we consider splitting the domain into $k \ge n^{\alpha / \delta}$ equally long intervals.
Let us choose $c_2$ as the endpoint of an arbitrary
interval $I$ and choose $s:= \lfloor n^{1 / \delta} \rfloor$, so
$ s^\alpha \le n^{\alpha / \delta} \le k $ and thus the above probability covers at least the whole interval $I$. That gives us:
\[	\Prob{X \in I} \lespaces \beta s^{-\delta}
	\eqspaces \beta {\lfloor n^{1 / \delta} \rfloor}^{-\delta}
	\lespaces \beta \left( n^{1 / \delta} - 1 \right) ^{-\delta}.
\] \[
\text{Since } \lim_{n \to \infty}
	\frac{ ( n^{1 / \delta} ) ^{-\delta} }{ ( n^{1 / \delta} - 1 ) ^{-\delta} }
	= 1 \text{, we conclude that } \Prob{X \in I} \text{ is } \OO(1/n).
\]
The input keys are chosen independently, so the number of keys in $I$ is given by the binomial distribution, and its expected value
$n \Prob{X \in I}$ is in $\OO(1)$.
\end{proof}

\begin{rem*}
We showed that the number of keys in an interval is expected to be constant, but the tail of that binomial distribution can be bounded even more tightly, e.\,g.~by Chernoff bounds~\cite[chapter~4.1]{randomAlgs}, to guarantee that high values are exponentially rare.
We do not elaborate on a finer analysis, because we feel a more pressing problem that we do not solve in this paper~-- the unrealistic (unmotivated) character of the used input model where all keys are inserted \emph{independently} from a (partially) known distribution.
We use those assumptions for our analysis to enable comparison with known structures, as we have only found one published structure that uses a different model~\cite{DemaineJP04} and moreover the bounds implied in that case seem relatively weak.
\end{rem*}


\section{Combining with known results and comparing to related work \label{sec:main-results}}

We propose to split the bit representation of every key into two parts,
exactly as one step of the decomposition from van Emde Boas trees,
except that we split asymmetrically. In both parts we propose to use
standard structures that do not utilize distribution properties of
the input. The speed on the more significant bits will be due to the
keys being short, and for the less significant bits the number of
keys will be small in expected case due to smoothness.

\subsection{Reviewing standard vEBT decomposition}

Standard van Emde Boas trees solve the predecessor problem with all
operations needing $\OO(\log l)$ time when working with keys of $l$
bits. That assumes $l$ is not asymptotically larger than the word
length of the used word-RAM.

The well-known decomposition can be viewed as performing binary search
for the longest prefix of bit-representation that is shared by the
searched key and some of the stored keys. The structure stores a mapping
from the more significant halves to corresponding subsets of the less
significant halves of the stored keys. These subsets are stored recursively
in the same way, and also the set of occurring more significant halves
is stored recursively. The operations on vEBT are carefully designed
to perform at most one nontrivial recursive call on each recursion
level, and all the rest are constant-time operations.

To keep the space linear\footnote{In word-RAM data structures, linear space means using $\OO(n)$ words
of memory.}, some modifications to the original structure are required. The mappings
can be represented by hash tables instead of simple arrays, and indirection
can be used: the whole set is split into $\Theta(\log l)$-sized consecutive
clusters in a fashion similar to B-tree nodes, and the actual vEBT construction
is only applied to the whole clusters joined together by a doubly
linked list. These changes make the complexity of insertions and deletions
expected and amortized where the expectation is only over random bits.

A possible way of doing these modifications is described in more details
e.g.~in~\cite[chapter 4]{Cunat10}, including the algorithms for
operations which can also be found in textbooks \cite[chapter 20]{intro2alg}.

\subsection{Modifying the vEBT decomposition}

To cover the cases most studied in previous work, we focus on
$\left(s^{\alpha},s^{1-\delta}\right)$-smooth distributions while using linear space.
To satisfy assumptions of the \ref{lem:smooth-bucket}, our top-level decomposition
needs to split away at least the most significant
$(\alpha/\delta)\log n$ bits.\footnote{In case there are not that many bits in the word, we can ``split
all of them'', meaning we basically use vEB trees directly while
fitting into the stated complexity bounds.} If the values of parameters are unknown, we can split away more bits
to be asympotically certain, e.g. $\log^{2}n$.

These bounds on the number of bits change during the life of the structure,
but we can easily ensure that we cut at least that many by performing
a full rebuild after every $n_{0}/4$ modifying operations where $n_{0}$
denotes the size of the stored set during the last rebuild. It would
also be possible to deamortize that process by standard technique
of global rebuilding \cite{global-rebuilding}. Also note that $\left(\alpha/\delta\right)\log n>\log^{2}n$
only for $n<2^{\alpha/\delta}$ which is an unknown constant; thus
the performance will be constant in that base case.

We can use standard hashed vEBT on the more significant bits to get
$\OO\left(\log\left(\log^{2}n\right)\right)=\OO\left(\log\log n\right)$
expected amortized time per operation. Each substructure for the less
significant bits is expected to store $\OO(1)$ keys. The resulting
performance is essentially the same as that of~Kaporis et~al.~\cite{KMSTTZ06},
except that unlike all previous work we achieve linear space without
requiring $\alpha\le1$, and our approach is conceptually much simpler.
\begin{thm*}
There is a structure for solving the predecessor problem in $\OO\left(\log\log n\right)$
expected time per operation and linear space on any $\left(s^{\alpha},s^{1-\delta}\right)$-smooth
distribution for arbitrary $\alpha,\delta>0$. The distribution and
parameters do not need to be known.
\end{thm*}


\subsection{Concluding remarks}

If we wanted to cut down time on unfavorable input, we could use the
same approach as~Kaporis et~al. There are standard structures for
the predecessor problem working in time $\OO(\!\sqrt{\log n/\log\log n})$
per operation on arbitrary input (considering the less precise bounds),
so these can be used for the less significant bits. Note that when
handling the more significant bits, the expected amortized time $\OO\left(\log\log n\right)$
remains independent of the input, as it only comes from randomization
during hashing.

We feel that it is good to decouple most of the analysis from complex
conditions on the input, as especially in practice we can rarely assume
that all operations on the structure are independent and identically
distributed. For the approach of our paper to work, it is enough (informally)
that the information is only carried by a sufficiently short bit prefix
of the keys.

\section*{Acknowledgements}

Most of the results presented above were included in the author's master thesis~\cite{Cunat10}. The author wishes to dedicate this paper to the memory of his supervisor V\'aclav Koubek, who passed away before this work was completed.

\bibliographystyle{plain} 
\bibliography{refs}

\end{document}